\documentclass[11pt,fleqn]{article}
\usepackage{geometry}
\usepackage{amsfonts}
\usepackage{amsmath}
\usepackage{graphicx}
\usepackage{bbm}
\usepackage{color}

\usepackage{latexsym}

\usepackage{amsthm}
\usepackage{a4wide}
\usepackage{amssymb}
\usepackage{url}
\usepackage{hyperref}

\newtheorem{Lemma}{Lemma}
\newcommand{\be}{\begin{equation}}
\newcommand{\ee}{\end{equation}}
\newcommand{\ber}{\begin{eqnarray}}
\newcommand{\eer}{\end{eqnarray}}
\newcommand{\curl}{{\rm curl}\,}
\newcommand{\di}{{\rm div}\,}

\begin{document}

\title{Localizing Energy in Fierz-Lanczos theory}
\author{
Jacek Jezierski\thanks{E-mail: Jacek.Jezierski@fuw.edu.pl} \, and Marian Wiatr\thanks{E-mail: mwiatr@fuw.edu.pl}\\
Department of Mathematical Methods in Physics, \\
Faculty of Physics, University of Warsaw,\\
ul. Pasteura 5, 02-093 Warsaw, Poland \\[1ex]
Jerzy Kijowski\thanks{E-mail: kijowski@cft.edu.pl}\\
Center for Theoretical Physics,\\
Polish Academy of Sciences, \\
Al. Lotników 32/46, 02-668 Warsaw, Poland}

\maketitle

\begin{abstract}
We calculate energy carried by the massless spin-2 field using Fierz-Lanczos representation of the theory. For this purpose Hamiltonian formulation of the field dynamic is thoroughly analyzed. Final expression for the energy is very much analogous to the Maxwell energy in electrodynamics (spin-1 field) and displays the locality property. Known as a
``super-energy'' in gravity theory, this quantity differs considerably from the well understood gravitational field energy (represented in linear gravity by the quadratic term in Taylor expansion of the A.D.M. mass) which cannot be localized.
\end{abstract}

\section{Introduction. Fierz-Lanczos field equations}

Linear gravity is a gauge-type field theory. The spacetime metric is split into a fixed ``background metric'' $g_{\mu\nu}$ and a ``small perturbation'' $h_{\mu\nu}$ playing a role of the configuration variable and admitting gauge transformations:
\begin{equation}\label{gauge}
    h_{\mu\nu} \longrightarrow h_{\mu\nu} + \pounds_\xi g_{\mu\nu} \, ,
\end{equation}
where the Lie derivative with respect to the vector field $\xi$ describes an ``infinitesimal coordinate transformation'' $x^\mu \rightarrow x^\mu + \xi^\mu(x)$. Linearized Einstein equations are second order differential equations imposed on the metric variable $h_{\mu\nu}$.

A substantial, technical simplification of the theory is obtained if we formulate it in terms of gauge-invariants. In case of the flat Minkowski background, an elegant gauge-invariant formulation is obtained in terms of components of the (linearized) Weyl tensor $W_{\lambda\mu\nu\kappa}$, i.e. the traceless part of the (linearized) curvature tensor
\begin{eqnarray}
 \label{Riemann}
  R^\lambda_{\kappa\mu\nu} &=& \nabla_\mu \Gamma^\lambda_{\kappa\nu} - \nabla_\nu \Gamma^\lambda_{\kappa\mu} \, ,
\end{eqnarray}
where $\Gamma$ represents the (linearized) connection coefficients of the total metric $g+h$:
\begin{equation}\label{Christofells}
    \Gamma^\lambda_{\mu\nu} = \frac 12 g^{\lambda\kappa} \left( h_{\kappa\mu ; \nu} + h_{\kappa\nu ; \mu} - h_{\mu\nu ; \kappa} \right) \, ,
\end{equation}
whereas both ``$\nabla$'' and ``;'' denote covariant derivative with respect to the background geometry $g$ (see e.g.
\cite{JJspin2}, \cite{JJnullweyl} and \cite{CYKCQG19}).

Due to metricity condition (\ref{Christofells}), Riemann tensor satisfies the following identities\footnote{Note that for tensors fulfilling \eqref{sym-asym}, identity \eqref{R-total-antisymm} is equivalent to first-type Bianchi identity:
\begin{equation}\label{Bianchi}
    R_{ \lambda[ \mu\nu\kappa ]} = 0 \, .
\end{equation}
}:
\begin{eqnarray}
  R_{\lambda\mu\nu\kappa} &=& - R_{\mu\lambda\nu\kappa} = - R_{\lambda\mu\kappa\nu} = R_{\nu\kappa\lambda\mu} \, , \label{R-sym-asym}\\
  R_{[ \lambda\mu\nu\kappa ]} &=& 0 \, .  \label{R-total-antisymm}
\end{eqnarray}
First identity leaves 21 independent components, so the Riemann tensor has 20 independent components. Half of them is carried by the Ricci tensor
\begin{equation}\label{Ricci}
    R_{\mu\nu} := R^\lambda_{\mu\lambda\nu}  \, ,
\end{equation}
which is symmetric (again -- due to metricity of the connection). Hence, the traceless part of the Riemann tensor:
\begin{equation}\label{Weyl}
    W_{\lambda\kappa\mu\nu} = R_{\lambda\kappa\mu\nu} - \frac 12 \left( g_{\lambda\mu} R_{\kappa\nu}
  - g_{\lambda\nu} R_{\kappa\mu} +  g_{\kappa\nu} R_{\lambda\mu} -  g_{\kappa\mu} R_{\lambda\nu}
  \right)  + \frac 16 R \left( g_{\lambda\mu} g_{\kappa\nu}
  - g_{\lambda\nu} g_{\kappa\mu}
  \right)\, ,
\end{equation}
called Weyl tensor, has 10 independent components. The complete list of its identities is:
\begin{eqnarray}
  W_{\lambda\mu\nu\kappa} &=& - W_{\mu\lambda\nu\kappa} = - W_{\lambda\mu\kappa\nu} = W_{\nu\kappa\lambda\mu} \, , \label{sym-asym}\\
  W_{[ \lambda\mu\nu\kappa ]} &=& 0 \, ,   \label{total-antisymm} \\
  W^\lambda_{\ \mu\lambda\kappa} &=& 0 \, . \label{slad-W=0}
\end{eqnarray}

It can be proved that the gauge-invariant content of linearized Einstein equations is equivalent to the ``contracted 2-nd type Bianchi'':
\begin{equation}\label{Bianchi-2}
    \nabla_\lambda W^{ \lambda \mu\nu\kappa } = 0 \, .
\end{equation}
In particular, the existence of the metric field $h_{\mu\nu}$, such that all the quantities arising here can be obtained by its appropriate differentiation, is guaranteed\footnote{More precisely, gauge-invariant part of vacuum Einstein metric $h$ is equivalent to spin-2 field $W$, see \cite{JJspin2} and Theorem 1 (formulae 2.15) in \cite{JJnullweyl}. However, one has to remember that the operator $h \mapsto W[h]$ has non-trivial kernel which includes `cosmological solutions'. Typical example (in spherical coordinates) is $h=r^2(dt^2+dr^2)$ which corresponds to linearized de Sitter metric. It gives $W[h]=0$ but its (linearized) Ricci is not vanishing.}  by (\ref{Bianchi-2}).

Spin-two-particle quantum mechanics can also be formulated in a similar language (cf.~\cite{Fierz}). Originally, the particle's ``wave function'' is described by the totally symmetric, fourth order spin-tensor. However, there is a one-to-one correspondence between such spin-tensors and tensors $W_{\lambda\mu\nu\kappa}$ satisfying identities (\ref{sym-asym}--\ref{slad-W=0})
(the transformation between the two pictures can, e.g., be found in \cite{Taub}). Moreover, evolution of a massless particle is governed by the same field equation (\ref{Bianchi-2}). In this representation, the theory is often referred to as the Fierz-Lanczos theory. Here, identities  (\ref{sym-asym}--\ref{slad-W=0}) are not treated as a consequence of any ``metricity'' (there is {\em a priori} no metric here!) but are a straightforward  consequence of the transformation from the spinorial to the tensorial language.

Fierz-Lanczos theory can also be derived from a variational principle and the corresponding ``potentials'' are known as Lanczos potentials \cite{Lanczos}--\cite{Baykal}. In the present paper we propose a substantial simplification of this theory on both the Lagrangian and the Hamiltonian levels. Finally, we calculate the field energy equal to the value of the field Hamiltonian and prove its local character. This means that if the region $V= V_1 \cup V_2$ is a union of two disjoint regions $V_1$ and $V_2$ then the corresponding field energies sum up:
\begin{equation}\label{Esum}
    E_V = E_{V_1} + E_{V_2} \, .
\end{equation}
Our main result is: the energy of the Fierz-Lanczos field\footnote{known in gravity theory as one of the so called ``super-energies''} is entirely different from the well understood (A.D.M.)-energy of the gravitational field. Linear expansion of the field dynamics in a neighbourhood of the background metric $g_{\mu\nu}$ corresponds to the quadratic expansion of the A.D.M. energy (``mass'') which has been calculated by Brill and Deser (see \cite{BrillDeser}).
Anticipating results which will be presented in the next paper, let us mention that gravitational energy cannot be localized: identity \eqref{Esum} cannot be valid in gravity theory because the gravitational interaction energy between the two energies (masses) has to be taken into account on the right-hand-side\footnote{This observation does not contradict the so called {\em quasi-localization} of gravitational energy.}.

We conclude that linear gravity and the Fierz-Lanczos theory differ considerably. They can be described by the same field $W$ and the same field equations (\ref{Bianchi-2}), but the corresponding phase spaces carry entirely different canonical (symplectic) structures. Consequently, energy carried by the field is entirely different in both theories. Graviton is not a simple ``massless spin-two particle''.

\section{Fierz-Lanczos field theory in (3+1)-formulation}

Quantum mechanics of a spin-two particle can be written either in the spinor or in the tensor language. The relation between the two equivalent formalisms can be found e.g. in the Taub paper \cite{Taub}. Here, we shall use the tensor formalism. This means that the field configuration is described by the ``Weyl-like'' tensor fulfilling identities (\ref{sym-asym}--\ref{slad-W=0}) typical for the  Weyl tensor of a metric connection.

In what follows we describe properties of the theory on a flat four-dimensional Minkowski space (signature $(-,+,+,+)$) whose metric $g_{\mu\nu}$ is used to rise and lower tensor indices.

Weyl-like tensor $W$ can be nicely described in a $(3+1)$-decomposition. Denoting by $t=x^0$ the time variable and by $(x^k), k=1,2,3$,  the remaining space variables\footnote{Here,  we use Lorentzian linear coordinates. Similarly as in Maxwell electrodynamics, generalization to curvilinear coordinates is obvious.}, 10 independent components of $W$ are uniquely described by two three-dimensional symmetric, traceless tensors (cf.~\cite{Cartin},\cite{CGQ32jjsm}):
\begin{equation}\label{D-oraz-H}
D^{kl} =  W^{0k0l} \, , \quad
 \quad {B}^{ji} = \frac 12 \varepsilon^{jkl}  W^{0i}{_{kl
 }}
\, .
\end{equation}
Trace $D^{ij}g_{ij}$ vanishes due to identity \eqref{slad-W=0}, whereas \eqref{Bianchi} implies vanishing of $B^{ij}g_{ij}$. Antisymmetric part of $B$ is given by $W^{0k}{_{kl}}$, so it vanishes because Weyl tensor is traceless.
In Cartesian coordinates components of the tensor density $\epsilon^{jkl}=\sqrt{\det \eta_{mn}}\,\varepsilon^{jkl}$ are equal to the corresponding  components of  the Levi-Civita tensor $\varepsilon^{jkl}= \epsilon^{jkl}/\sqrt{\det \eta_{kl}}$ because $\det \eta_{kl} =1$.

Field equations $\nabla_\lambda W^{\lambda \mu\nu\kappa}=0$ can be written in a way similar to Maxwell electrodynamics:
\begin{eqnarray}
  {\rm div} D  &=& 0 \, , \label{dD}\\
  {\rm div} B &=& 0 \, , \label{dB}  \\
  \dot{D} &=& {\rm curl}\, B \, , \label{rB}\\
  \dot{B} &=&  - {\rm curl} \, D  \label{rD} \, .
\end{eqnarray}
where ``dot'' denotes the time derivative $\partial_0$. Moreover, the following differential operators of rank 1, acting on symmetric, traceless tensor fields $K^{ij}$ have been introduced:
\begin{eqnarray}
  \left({\rm div} K\right)_l &=& \nabla_k K^k_{\ l} \, , \label{div} \\
  \left( {\rm curl}\, K \right)_{ij} &=& \frac 12 \left( \varepsilon_i^{\ kl}\nabla_k K_{l j} + \varepsilon_j^{\ kl}\nabla_k K_{l i} \right)
  = \nabla_k K_{l (j} \varepsilon_{i)}^{\ kl} \label{rot}\, .
\end{eqnarray}
It is obvious that ${\rm curl} \, K$ is also a symmetric, traceless tensor.

For {\em transverse-traceless} tensors $D$ i $B$ (i.e. fulfilling constrains  (\ref{dD}--\ref{dB})), symmetrization in formula \eqref{rot} is not necessary because the antisymmetric part of $\varepsilon_i^{\ kl}\nabla_k K_{l j}$ vanishes:
\begin{equation}\label{div-curl}
    \varepsilon^{nij} \varepsilon_i^{\ kl}\nabla_k K_{l j}  = \varepsilon^{ijn} \varepsilon_i^{\ kl}\nabla_k K_{l j} =
    \left( g^{jk} g^{nl} - g^{jl} g^{nk} \right)\nabla_k K_{l j} = \nabla_k K^{nk} - \nabla^n K_j^{\ j}= 0 \, .
\end{equation}

\section{A simple variational principle (not obeying Lorentz-invariance)}\label{simple}

Similarly as in electrodynamics, field equations (\ref{dD}--\ref{rD}) can be derived from a variational principle. For this purpose we use the following simple observation (see Appendix for an easy proof):

\noindent
{\bf Lemma:} Given a symmetric, transverse-traceless field $B$ on a 3D-Euclidean space (i.e.~the Cauchy surface $\{ t = 0 \}$), there is a symmetric, transverse-traceless field $p$ such that
\begin{equation}\label{B=rotp}
    B = {\rm curl} \ p \, .
\end{equation}
The field $p$ is unique up to second derivatives $\partial_i \partial_j \varphi$ of a harmonic function: $\Delta \varphi =0$.

\noindent
{\bf Corollary:} Given field configuration $(D,B)$ satisfying field equations (\ref{dD}--\ref{rD}) on Minkowski spacetime $M$, there is a symmetric, transverse-traceless field $p$ on each Cauchy hypersurface $\{  t = {\rm const.} \}$ which fulfills not only (\ref{B=rotp}) but, moreover,
\begin{equation}\label{DpTS1}
    D = - \dot{p}  \, .
\end{equation}
The field $p$ satisfies wave equation
\begin{equation}\label{wave}
    \ddot{p} = \Delta p \, .
\end{equation}

\noindent
\begin{proof}
At each hypersurface $\{ t = {\rm const.} \}$ choose any $\widetilde{p}$ satisfying (\ref{B=rotp}). Due to field equations we have:
\[
    {\rm curl} \left( D + \dot{\widetilde{p}} \right) = {\rm curl} D + \dot{B} = 0 \, .
\]
Hence, at each instant of time $(D+\dot{\widetilde{p}})$ differs from zero by $\partial_i \partial_j \varphi$, where $\Delta \varphi =0$.
Integrating with respect to time, we can find $\alpha$ such that $\dot{\alpha} = \varphi$ and $\Delta \alpha =0$. Whence:
\[
    D + \dot{\widetilde{p}} = \partial_i \partial_j \dot{\alpha} \, .
\]
We conclude that
\begin{equation}\label{tildep}
    p := \widetilde{p}  - \partial_i \partial_j \alpha
\end{equation}
fulfills (\ref{DpTS1}). Taking into account that ${\rm curl}\, {\rm curl} = - \Delta$ on symmetric, transverse-traceless fields, we obtain:
\[
    \ddot{p} = - \dot{D} =  - {\rm curl}\, B =  - {\rm curl} \  {\rm curl} \, p = \Delta p \, . 
\] \end{proof}

\noindent
{\bf Remark:} The object $p$ is analogous to the vector potential $A_k$ in electrodynamics. Condition ${\rm div}\, p = 0$ plays a role of the Coulomb gauge. Condition (\ref{DpTS1}) plays a role of the additional axial gauge $A_0 = 0$, which can always be imposed on the Coulomb gauge.

Similarly as in electrodynamics, we can assume that the first pair of ``Maxwell equations'' is satisfied {\em a priori} and derive the remaining equations from a variational principle. For this purpose we treat $p$ as a field potential, equations (\ref{B=rotp}) and (\ref{DpTS1}) as definition of $D$ and $B$, and take the following Lagrangian function\footnote{A constant $\alpha$ is necessary because, contrary to the case of electrodynamics, the quantity $\frac{D^2 - B^2}2$ does not carry correct physical units. Actually, $\alpha$ must be calculated in $\ell^2$-units -- just an inverse to the cosmological constant units. The physically correct value of $\alpha$ can be measured if we know how the field $W$ interacts with any realistic field theory. Of course, the dimensional constant $\alpha$ could also be integrated {\em a priori} into definition of the fields $D$ and $B$, but then $W$ would not have the correct dimension of the curvature.}:
\begin{equation}\label{varia-p}
    {\cal L} (p, \dot{p}) := \alpha \cdot \frac {D^2 - B^2}2 \, .
\end{equation}

Indeed, we have:
\begin{equation}\label{dL-curl}
    \delta \int {\cal L} = \alpha \int \left( \dot{p} \delta \dot{p} -  ({\rm curl}\, p) \delta ({\rm curl}\, p) \right) = \alpha \int \left( -D \delta \dot{p} + \Delta p \delta p \right) \, ,
\end{equation}
which implies (\ref{wave}) as the Euler-Lagrange equation for $L$. Moreover, quantity $-\alpha D = \frac{\partial {\cal L}}{\partial \dot{p}}$ plays a role of the momentum canonically conjugate to $p$. To simplify notation, we shall skip the constant $\alpha$ in what follows (e.g., using appropriate physical units in which $\alpha = 1$).

Formula (\ref{varia-p}) implies the following Hamiltonian density of the field:
\begin{equation}\label{Hamilt}
    {\cal H} := (- D) \dot{p} - {\cal L} = D^2 - \frac {D^2 - B^2}2 = \frac {D^2 + B^2}2 \, ,
\end{equation}
which generates the Hamiltonian field dynamics
\[
    -\dot{p} =  \frac{\delta {\cal H}}{\partial (-D)}  \ \ \ ; \ \ \  - \dot{D} = \frac{\delta {\cal H}}{\partial p}
\] according to:
\begin{equation}\label{H-gen}
    \delta {\cal H} = D \delta D + B \delta ({\rm curl}\, p) = - \dot{p} \delta (- D) - ({\rm curl}\, B) \delta p + \{ {\rm boundary \ \ terms} \}
\end{equation}

\noindent
{\bf Remark:} Quantity
\begin{equation}\label{E_V}
    E_V := \int_V {\cal H}
\end{equation}
may be identified with amount of the field energy contained in $V$, provided the boundary term vanishes when integrating \eqref{H-gen} over $\partial V$. For this purpose appropriate boundary conditions have to be imposed (cf. \cite{KIJ}). Physically, control of boundary data ensures adiabatical insulation of the interior of $V$ from its exterior. From the functional-analytic point of view boundary conditions are necessary for the self-adjointness of the evolution operator (the Laplacian $\Delta$ in our case) which guarantees the existence and uniqueness of the Cauchy problem\footnote{The issue of energy localization will be thoroughly discussed in the next paper. Here, we limit ourselves to discussion of the strongest possible boundary conditions: all the fields vanish in a neighbourhood of the boundary $\partial V$. This condition annihilates all the surface integrals arising during integration by parts. Consequently, the Laplacian operator $\Delta$ arising here is a symmetric operator. In order to have field evolution correctly defined, its appropriate self-adjoint extension has to be defined. For this purpose, correct boundary conditions are necessary. In case of the
{\em total} field energy (i.e. when $V = \mathbb{R}^3$), boundary terms vanish due to the sufficiently fast fall-of behaviour of the field. Anticipating those results let us mention that, similarly to electrodynamics, the spin-two-particle theory admits the energy localization and the quantity (\ref{Hamilt}) is a correct {\em local} energy density, whereas linear gravity {\em does not} admit localization of energy.}  within $V$.

Hamiltonian description of the field evolution leads, therefore, to the phase space of initial data  parameterized by the configuration $p$ and the canonical momentum $-D$. This means that the space carries the following symplectic structure:
\begin{equation}\label{Omega_pD}
    \Omega = \int_V \delta p \wedge \delta D \, ,
\end{equation}
and the Hamiltonian \eqref{E_V} generates field dynamics \eqref{DpTS1} -- \eqref{wave}.

Being correct from the Hamiltonian point of view, above Lagrangian version of the theory is not satisfactory because it is not relativistic invariant. Indeed, field equations (\ref{Bianchi-2}) are relativistically invariant. Lorentz transformations of $W_{\lambda\kappa\mu\nu}$ uniquely imply transformation laws for $D$ and $B$. But, like in electrodynamics, transformation law for the ``Coulomb-gauged'' potential $p$ is not only non-relativistic but obviously non-local. In electrodynamics, Lorentz transformations can be  applied correctly to the four-potential $A_\mu$. They mix different gauges. Here, one could relax the Coulomb gauge ${\rm div} \, p = 0$ by adding a ``symmetric-traceless part of a gradient'', namely:
\begin{equation}\label{TS}
    TS(\nabla b)_{ij} := \frac12(\partial_i b_j + \partial j b_i) - \frac13 g_{ij} \partial_k b^k \, ,
\end{equation}
where $b$ is a three-vector field. This would be an analog of the ``gradient gauge'' $\partial_k \varphi$ in electrodynamics which can be added to $A_k$ without changing the field $B$. If, moreover, we add $\dot{\varphi}$ to $A_0$, also the field $D$ does not change. Unfortunately, here only divergence-free fields $\partial_k b^k =0$ can be used in (\ref{TS}) if we want to keep equation ${\rm curl}\, p = B$. Such a non-relativistic condition does not allow us to organize both $p$ and $b$ into a single, local, fully relativistic object.

The unique remedy for this disease which exists in the literature is the use of the so called Lanczos potentials, i.e. further relaxation of  (\ref{B=rotp}) and (\ref{DpTS1}).

\section{Lanczos potentials and the relativistic invariant variational principle}

Since Weyl tensor is obtained by differentiating connection coefficients $\Gamma^\lambda_{\mu\nu}$, they are natural candidates for potentials describing Lanczos field. But -- contrary to linear gravity -- there is  {\em a priori} no metric $h$ here. Hence, what we obtain by this procedure from a generic connection:
\begin{equation}\label{Riem-down}
    R_{\lambda\kappa\mu\nu} = - \Gamma_{\lambda\kappa\mu ; \nu} + \Gamma_{\lambda\kappa\nu ; \mu}
\end{equation}
does not satisfy symmetry conditions (\ref{sym-asym}) (to simplify further considerations we have lowered first index of the connection: $\Gamma_{\lambda\mu\nu} = g_{\lambda\sigma} \Gamma^\sigma{_{\mu\nu}}$). To produce Lanczos field we must use appropriate symmetrization:
\begin{equation}\label{Riem-down-symm}
    r_{\lambda\kappa\mu\nu} := R_{[\lambda\kappa]\mu\nu} + R_{[\mu\nu]\lambda\kappa} \, ,
\end{equation}
and finally eliminate traces:
\begin{equation} \label{weylA}
w_{{\alpha\beta}\mu\nu} := r_{{\alpha\beta}\mu\nu} -\frac12 \left(
r_{\alpha\mu}\eta_{\beta\nu}-r_{\alpha\nu}\eta_{\beta\mu} +\eta_{\alpha\mu}r_{\beta\nu}-\eta_{\alpha\nu}r_{\beta\mu} \right)
+ \frac16 (\eta_{\alpha\mu}\eta_{\beta\nu}-\eta_{\alpha\nu}\eta_{\beta\mu})r \, ,
\end{equation}
where we denoted:
\begin{equation} \label{riemA}  r_{\alpha\beta}= r^\mu{_{\alpha\mu\beta}}
\, , \quad r=r_{\mu\nu}\eta^{\mu\nu} \, .
\end{equation}
This object fulfills already identities (\ref{sym-asym}--\ref{slad-W=0}) i.e. is a genuine Fierz-Lanczos field.

Decomposing $\Gamma_{\lambda\mu\nu}$ into irreducible parts, we see that only one of them enters into definition \eqref{weylA} of $w$. Taking into account its symmetry:  $\Gamma_{\lambda\mu\nu}=\Gamma_{\lambda(\mu\nu)}$, we first decompose it into the totally symmetric part and the remaining part whose totally symmetric part vanishes:
\begin{equation}\label{decomp1}
    \Gamma_{\lambda\mu\nu} = \Gamma_{(\lambda\mu\nu)} + \widetilde{\Gamma}_{\lambda\mu\nu} \, ,
\end{equation}
with $\widetilde{\Gamma}_{(\lambda\mu\nu)} = 0$. This way 40 independent components of $\Gamma$ split into 20 components of the totally symmetric, rank 3 tensor and the remaining 20 components of $\widetilde{\Gamma}$. The first part drops out from (\ref{Riem-down}).

Instead of  $\widetilde{\Gamma}$, in most papers devoted to Lanczos potentials, the authors use its antisymmetrization in first indices:
\begin{equation}\label{tildeA}
    \widetilde{A}_{\lambda\mu\nu} := \widetilde{\Gamma}_{[\lambda\mu]\nu} \, .
\end{equation}
Vanishing of the totally symmetric part of $\widetilde{\Gamma}$ implies vanishing of the totally antisymmetric part of the new object: $\widetilde{A}_{[\lambda\mu\nu]}=0$. We stress, however, that both objects are equivalent: no information is lost during such an antisymmetrization, because there is a canonical isomorphism between  both types of tensors. Indeed, it is easy to check that the inverse transformation (from  $\widetilde{A}$ to $\widetilde{\Gamma}$) is given by the symmetrization operator:
\begin{equation}\label{inv-tildeA}
    \widetilde{\Gamma}_{\lambda\mu\nu} = \frac 34 \widetilde{A}_{\lambda(\mu\nu)} \, .
\end{equation}
We see that (\ref{Riem-down}) and (\ref{Riem-down-symm}) imply:
\begin{equation}
    r_{\lambda\kappa\mu\nu} := - \widetilde{A}_{\lambda\kappa\mu ; \nu} + \widetilde{A}_{\lambda\kappa\nu ; \mu}
    - \widetilde{A}_{\mu\nu\lambda ; \kappa} + \widetilde{A}_{\mu\nu\kappa ; \lambda}  \, .
\end{equation}
Finally, when passing to the Fierz-Lanczos field  (\ref{weylA}), the trace $\widetilde{A}_\lambda:=\widetilde{A}_{\lambda\mu\nu}g^{\mu\nu}$ drops out. Hence, we define the Lanczos potential as the traceless part of $\widetilde{A}$:
\begin{equation}\label{A-final}
    A_{\lambda\mu\nu}:= \widetilde{A}_{\lambda\mu\nu} - \frac 13\left( \widetilde{A}_\lambda g_{\mu\nu}
    - \widetilde{A}_\mu g_{\lambda\nu} \right)  \, .
\end{equation}
This object fulfills the following algebraic identities:
\begin{eqnarray}
  {A}_{ \lambda\mu\nu } &=& - {A}_{\mu \lambda\nu } \label{A-sym} \, ,\\
  {A}_{[ \lambda\mu\nu ]} &=& 0 \, ,\label{A-Bianchi}\\
  {A}_{\lambda\mu}{^{\mu}} &=& 0  \,  \label{slad=0}
\end{eqnarray}
(see also \cite{Edgar} and \cite{BampiCaviglia}). It has 16 independent components, because 4 among the original 20 was carried by the trace $\widetilde{A}_\lambda$.

The field $w$ written explicitly in terms of $A$ looks as follows (see \cite{Edgar}):
\begin{equation} \label{wodA}
w_{{\alpha\beta}\mu\nu} = 2A_{{\alpha\beta}[\nu ; \mu]}
+ 2A_{{\nu\mu}[\alpha ; \beta]} - (A^\sigma{_{(\alpha\mu);\sigma}}\eta_{\beta\nu} - A^\sigma{_{(\alpha\nu);\sigma}}\eta_{\beta\mu} +
A^\sigma{_{(\beta\nu);\sigma}}\eta_{\alpha\mu} - A^\sigma{_{(\beta\mu);\sigma}}\eta_{\alpha\nu}) \, .
\end{equation}

Let $\widetilde{\Gamma}_\lambda := \widetilde{\Gamma}_{\lambda\mu}{^\mu}$. Observe that $\gamma_{ \lambda\mu\nu }$ defined as the traceless part of $\widetilde{\Gamma}_{\lambda\mu\nu}$:
\[
    \gamma_{ \lambda\mu\nu } = \widetilde{\Gamma}_{\lambda\mu\nu} - \frac 13 \left(\widetilde{\Gamma}_\lambda g_{\mu\nu} - \widetilde{\Gamma}_{(\mu} g_{\nu) \lambda}     \right) \, ,
\]
contains the same information as $A_{\lambda\mu\nu}$:
\begin{equation}\label{gammal}
    A_{\lambda\mu\nu} = \gamma_{[\lambda\mu]\nu} \, ; \quad  \gamma_{\lambda\mu\nu} = \frac 34 A_{\lambda(\mu\nu)} \, .
\end{equation}
This object fulfills the following algebraic identities:
\begin{eqnarray}
  \gamma_{ \lambda\mu\nu } &=&  \gamma_{\lambda\nu\mu} \label{gamma-sym} \, ,\\
  \gamma_{( \lambda\mu\nu )} &=& 0 \, ,\label{gamma-Bianchi}\\
  \gamma_{\lambda\mu}{^{\mu}} &=& 0  \, , \label{gamma-slad=0}
\end{eqnarray}
and the corresponding expression for the Fierz-Lanczos field reads:
\begin{equation}
w_{{\alpha\beta}\mu\nu} = 2\gamma_{{[\alpha\beta]}[\nu ; \mu]}
+ 2\gamma_{{[\nu\mu]}[\alpha ; \beta]} - \frac 34\left( \gamma^\sigma{_{\alpha\mu;\sigma}}\eta_{\beta\nu} - \gamma^\sigma{_{\alpha\nu;\sigma}}\eta_{\beta\mu} +
\gamma^\sigma{_{\beta\nu;\sigma}}\eta_{\alpha\mu} - \gamma^\sigma{_{\beta\mu ;\sigma}}\eta_{\alpha\nu}\right) \, .
\end{equation}
Hence, there are two equivalent versions of potentials for the Fierz-Lanczos field.
In what follows, we shall use $A_{\lambda\mu\nu}$ -- the version proposed by Lanczos, as being more popular in the literature.

\section{A relativistic variational principle for Fierz-Lanczos theory}

Take an invariant Lagrangian density $L = L(w)$. It depends upon potentials and its first derivatives {\em via} $w$, exclusively. Euler-Lagrange'a equations
\begin{equation}\label{E-L}
    \frac {\delta L}{\delta A_{\lambda\mu\nu}} = 0 \,
\end{equation}
can be written in a ``symplectic'' way
\begin{equation}\label{E-L-sympl}
    \delta L (A, \partial A) = \partial_\kappa \left( {\cal W}^{\lambda\mu\nu\kappa}\delta A_{\lambda\mu\nu} \right) =
    \left(\partial_\kappa {\cal W}^{\lambda\mu\nu\kappa} \right) \delta A_{\lambda\mu\nu} + {\cal W}^{\lambda\mu\nu\kappa} \delta A_{\lambda\mu\nu , \kappa} \, ,
\end{equation}
or, equivalently:
\begin{eqnarray}
 \partial_\kappa {\cal W}^{\lambda\mu\nu\kappa} &=& \frac {\partial L}{\partial A_{\lambda\mu\nu}} \, ,\\
 {\cal W}^{\lambda\mu\nu\kappa} &=& \frac {\partial L}{\partial A_{\lambda\mu\nu , \kappa}} \, .
\end{eqnarray}
Canonical momentum ${\cal W}$ is a tensor density, because $L$ was a scalar density and we can equivalently use tensor $W$, such that ${\cal W} = \sqrt{|\det g |} W$.
These equations can be formulated in a covariant form. We observe for this purpose, that expression ${\cal W}^{\lambda\mu\nu\kappa}\delta A_{\lambda\mu\nu}$ is a vector density, so its (partial) divergence is equal to covariant divergence. Therefore, equation \eqref{E-L-sympl} can be rewritten:
\begin{equation}\label{E-L-sympl-cov}
    \delta L (A, \partial A) = \nabla_\kappa \left( {\cal W}^{\lambda\mu\nu\kappa}\delta A_{\lambda\mu\nu} \right) =
    \left(\nabla_\kappa {\cal W}^{\lambda\mu\nu\kappa} \right) \delta A_{\lambda\mu\nu} + {\cal W}^{\lambda\mu\nu\kappa} \delta A_{\lambda\mu\nu ; \kappa} \, .
\end{equation}
But $L$ does not contain components of $A$ \textit{explicite} but only covariant derivatives of $A$. Hence, we obtain field equations:
\begin{eqnarray}
  \nabla_\kappa {\cal W}^{\lambda\mu\nu\kappa} &=& 0 \label{rPola1} \, , \\
  {\cal W}^{\lambda\mu\nu\kappa} &=& \frac {\partial L}{\partial A_{\lambda\mu\nu ; \kappa}} \, . \label{rPola2}
\end{eqnarray}
First equation is universal, but relation between $w$ and its momentum ${\cal W}$ is implied by a specific form of the Lagrangian. Define derivative of $L$ with respect to $w$ by the following identity:
\begin{equation}\label{def_w}
     \delta L = \frac {\partial L}{\partial w_{\lambda\mu\nu \kappa}} \  \delta w_{\lambda\mu\nu \kappa} \,.
\end{equation}
The quantity $\frac {\partial L}{\partial w_{\lambda\mu\nu \kappa}}$ belongs to the (vector) space of contravariant tensor densities. Due to the spacetime metric $g$, it is equipped with the (pseudo-)Euclidean, non-degenerate structure. Splitting  this vector space into a direct sum of tensors having the same symmetries as the Weyl tensor and its orthogonal complement (we denote by $P_w$ and $P^\perp_w$, respectively, the corresponding projections), we write
\begin{equation}
    \frac {\partial L}{\partial w_{\lambda\mu\nu \kappa}} = P_w\left(\frac {\partial L}{\partial w_{\lambda\mu\nu \kappa}}\right) + P^\perp_w\left(\frac {\partial L}{\partial w_{\lambda\mu\nu \kappa}}\right)
\end{equation}
and, consequently,
\begin{equation}
     \delta L = \left[ P_w\left(\frac {\partial L}{\partial w_{\lambda\mu\nu \kappa}}\right) + P^\perp_w\left(\frac {\partial L}{\partial w_{\lambda\mu\nu \kappa}}\right)\right] \delta w_{\lambda\mu\nu \kappa}  =  P_w\left(\frac {\partial L}{\partial w_{\lambda\mu\nu \kappa}}\right)\delta w_{\lambda\mu\nu \kappa} \,.
\end{equation}
We see that condition $\frac {\partial L}{\partial w_{\lambda\mu\nu \kappa}} = P_w\left(\frac {\partial L}{\partial w_{\lambda\mu\nu \kappa}}\right)$ is necessary to give an unambiguous meaning to the definition (\ref{def_w}): it must fulfil the same algebraic identities as $w$ does. Whence:
\[
    \delta L = \frac {\partial L}{\partial w_{\lambda\mu\nu \kappa}} \  \delta w_{\lambda\mu\nu \kappa}= \frac {\partial L}{\partial w_{\lambda\mu\nu \kappa}} \  \delta r_{\lambda\mu\nu \kappa} = 4\  \frac {\partial L}{\partial w_{\lambda\mu\nu \kappa}} \ \delta A_{\lambda\mu\kappa ; \nu} \, ,
\]
which means that:
\begin{equation}\label{pole3}
   {\cal W}^{\lambda\mu\kappa\nu} = 4\  \frac {\partial L}{\partial w_{\lambda\mu\nu\kappa}} \, .
\end{equation}
Taking (cf.~\cite{Cartin})
\begin{equation}\label{LP-F}
    L = \frac1{16} \sqrt{|\det g |} w^{\lambda\mu\nu \kappa} w_{\lambda\mu\nu \kappa}
\end{equation}
we obtain
\[
    \delta L = \frac 18 \sqrt{|\det g |} \ w^{\lambda\mu\nu \kappa} \ \delta w_{\lambda\mu\nu \kappa}
    =  \frac 12 \sqrt{|\det g |} \ w^{\lambda\mu\nu \kappa} \ \delta A_{\lambda\mu\kappa ; \nu} \, ,
\]
so finally:
\begin{equation}\label{pole4}
    {\cal W}^{\lambda\mu\nu\kappa} = - {\cal W}^{\lambda\mu\kappa\nu} = -\frac 12 \sqrt{|\det g |} \, w^{\lambda\mu\nu \kappa} \, .
\end{equation}

\section{(3+1)-decomposition of the Lanczos potentials. Analogy with electrodynamics}

In (3+1)-decomposition the ``velocity tensor'' $w$ can be represented by two $3D$ symmetric, traceless tensors\footnote{For simplicity, we restrict ourselves to the flat case. This means that the Cauchy surface $\{ t = {\rm const.}\}$ carries the flat Euclidean metric $\eta_{kl}$ and we use Cartesian coordinates. Consequently, components of the tensor density $\epsilon^{jkl}$ are equal to the corresponding  components of the Levi-Civita tensor $\varepsilon^{jkl}= \epsilon^{jkl}/\sqrt{\det \eta_{kl}}$ since $\det \eta_{kl} =1$. Generalization to the curved space is relatively straightforward.}, which we call $E$ and $B$:
\begin{equation}\label{E-oraz-B}
E_{kl} =  w_{0k0l} \, , \quad
 \quad {B}^{ji} = \frac 12 \epsilon^{jkl}  w^{0i}{_{kl}}
\, .
\end{equation}
In analogy with electrodynamics, the corresponding components\footnote{Introducing ${\cal F}^{\lambda\mu\nu\kappa}:= -2{\cal W}^{\lambda\mu\nu\kappa}= \sqrt{|\det g |} \, w^{\lambda\mu\nu \kappa}$ we can define $D$, $H$ in a way analogous to (\ref{E-oraz-B}): $D^{kl} := {\cal F}^{0k0l}$ and $H_{kl}:= \frac12 \epsilon_{kij} {\cal F}^0{_l}^{ij}$.}
of the ``momentum tensor'' $W$ could be called  $D$ and $H$ (cf. \eqref{D-oraz-H}), but the Lagrangian \eqref{LP-F} implies the ``constitutive equations'' \eqref{pole4} equivalent to: $D=E$, $H=B$.
It is easy to show (proof in the Appendix), that
\begin{equation}\label{w-DB}
    w_{\lambda\mu\nu\kappa} w^{\lambda\mu\nu\kappa} = 8 \left( E^2 - B^2 \right) \Longrightarrow
    L= \frac 12 \sqrt{|\det g |} \left( E^2 - B^2 \right) \, .
\end{equation}

The Lanczos potential $A$, which has 16 independent components, splits into two symmetric, traceless, three-dimensional tensors $P_{ij}$ and $S_{ij}$ and two three-dimensional covectors $a_i$ and $b_i$. The latter are defined {\em via} decomposition of the three-dimensional two-form $A_{ij0}$:
\begin{eqnarray}
  a_i &=& -A{_{0i0}} \, , \\
  b^i &=& -\frac 12 \varepsilon^{ikl} A_{kl0}   \ \  \Leftrightarrow  \ \   A_{ij0}=-b^m\varepsilon_{mij} \, ,
\end{eqnarray}
whereas $P$ and $S$ are defined  as a symmetric part of $A_{0kl}$ and $A_{ijk}\varepsilon^{ij}{_l}$, respectively. Antisymmetric parts of them are already given by $a$ and $b$, due to identities fulfilled by $A$. More precisely, we have (proof in the Appendix):
\begin{eqnarray}
\label{A0kl}
  A_{0kl} &=& -2P_{kl}+\frac12 b_j \varepsilon^j{_{kl}} \, , \\
 \label{Aijk}
 \frac12 A_{ijk}\varepsilon^{ij}{_l} &=& -2S_{kl} + \frac12 a_j \varepsilon^j{_{kl}} \ \  \Leftrightarrow  \ \  A_{ijk} =    -2S_{kl}\varepsilon^l{_{ij}}+ \frac12(a_i\eta_{jk}-a_j\eta_{ik}) \, .
\end{eqnarray}
Relation \eqref{weylA} between potentials  $A$ and the field $w$ can be written in terms of these three-dimensional objects. We obtain (proof in the Appendix):
\begin{eqnarray}
  E_{kl}= w_{0k0l} &=&  - \partial_0 P_{kl} + \partial_i S_{j(k}\varepsilon_{l)}{^{ij}}
 + \frac34 (\partial_l a_{k}+\partial_k a_{l}) - \frac12 \eta_{kl}\partial_i a{^i} \, ,  \\
   B_{kl}= \frac12 \varepsilon^{ij}{_{l}} w_{k0ij} &=&  \partial_0 S_{kl} + \partial_i P_{j(k}\varepsilon_{l)}^{\ \ ij}
-\frac34 (\partial_l b_{k} + \partial_k b_{l})+\frac12 \eta_{kl}\partial_i b^i \, .
\end{eqnarray}
These relations can be written shortly as:
\begin{equation}
\label{weylodA}
 E =-\dot P +{\rm curl} \, S + \frac32TS(\nabla a)   \, , \quad   B = \dot S + {\rm curl} \,  P -\frac32 TS(\nabla b) \, ,
\end{equation}
where by ``$TS(\nabla b)$'' we denote the {\em  traceless, symmetric part} of $\nabla b$. Hence, in Lorentzian coordinates, Lagrangian density of the theory can be expressed in terms of potentials as:
\begin{eqnarray}
  L &=& \frac 1{16} \sqrt{|\det g |} w_{\lambda\mu\nu\kappa} w^{\lambda\mu\nu\kappa} = \frac12 \sqrt{|\det g |} \left( E^2 - B^2 \right)  \\
    &=& \frac 12  \left\{ \left( \dot P - {\rm curl} \, S - \frac32 TS(\nabla a)\right)^2 - \left(  \dot S + {\rm curl} \,  P -\frac32 TS(\nabla b)\right)^2 \right\} \label{Lagr1} \, .
\end{eqnarray}
We see, that constraints (\ref{dD}--\ref{dB}) are obtained from variation of $L$ with respect to $a$ and $b$, whereas dynamical equations (\ref{rB}--\ref{rD}) from variation with respect to $P$ and $S$. This equations expressed by potentials $(P,S,a,b)$ have the following form:
\begin{align}\label{potP}
\frac32TS\left(\nabla(\dot a + \frac12\curl b)\right) &= \ddot P+ \curl\curl P \, , \\ \label{potS}
  \frac32TS\left(\nabla(\dot b - \frac12\curl a)\right) &= \ddot S+ \curl\curl S \, .
\end{align}

\section{Fierz-Lanczos formulation of Maxwell electrodynamics}

In (3+1)-decomposition, Fierz-Lanczos theory shows a far reaching analogy with electrodynamics. The only difference is that in FL theory we have two ``vector potentials'' ($P$ and $S$) instead of one ($A_k$) in electrodynamics, and two ``scalar potentials'' ($a$ and $b$) instead of one ($A_0$) in electrodynamics.  To clarify this structure, we show in this Section how to formulate here classical electrodynamics in a similar way, i.e. using two independent potentials.

Conventionally, classical (linear or non-linear) electrodynamical field is described by two differential two-forms: $f = f_{\mu\nu} {\rm d} x^\mu \wedge {\rm d} x^\nu$ and ${\cal F} = \frac 12 {\cal F}^{\mu\nu} \epsilon_{\mu\nu\alpha\beta} {\rm d} x^\alpha \wedge {\rm d} x^\beta$. First pair of Maxwell equations: ${\rm d}f = 0$ and the second pair: ${\rm d} {\cal F} = J$ are universal, whereas ``constitutive equations'', i.e. relation between $f$ and ${\cal F}$ depends upon a model. In particular, linear Maxwell theory corresponds to the relation ${\cal F} = *f$, where by ``$*$'' we denote the Hodge ``star operator''.

Usually, we derive the theory from the variational principle, where the first pair of Maxwell equations is assumed {\em a priori}. For this purpose we substitute: $f = {\rm d} A$, or
\[
f_{\mu\nu}= \partial_\mu A_\nu - \partial_\nu A_\mu=
A_{\nu , \mu} - A_{\mu , \nu}
 \]
in coordinate notation, where $A = (A_\mu)$ is a four-potential one-form and $A_{\nu , \mu}:= \partial_\mu A_\nu$. In $(3+1)$-decomposition, electric and magnetic fields are then defined by components of $f$:
\begin{equation} \label{polaEB}
(f_{k0}) = \vec{E}   = - \dot {\vec{A}}   +\vec{\nabla} A_0 \, , \quad \frac 12 \left( \epsilon^{mkl} f_{kl}\right) = \vec{B}  =   {\rm curl} \vec{A}     \, ,
\end{equation}
whereas inductions: $\vec{D}$ and $\vec{H}$ arise as corresponding canonical momenta. More precisely, variational principle can be written as follows:
\begin{equation}
\delta L(A_\nu , A_{\nu , \mu}) = \partial_\mu ({\cal F}^{\nu\mu} \delta
A_\nu)  =
(\partial_\mu {\cal F}^{\nu\mu}) \delta A_\nu +
{\cal F}^{\nu\mu} \delta A_{\nu , \mu} \ , \label{deltaL-Elmag}
\end{equation}
equivalent to
\begin{eqnarray}
\partial_\mu {\cal F}^{\nu\mu}  =  \frac {\partial L}{\partial A_\nu} = J^\nu  \; ,
\qquad
{\cal F}^{\nu\mu}  =  \frac {\partial L}{\partial A_{\nu , \mu}}
= 2 \frac {\partial L}{\partial f_{\mu\nu}}
\, ,
\label{E-LpoleElmag}
\end{eqnarray}
where the components of the canonical momentum tensor ${\cal F}$ are:
\begin{equation}\label{HD}
 {\cal F}^{0k}=- {\cal F}^{k0}=  {\cal D}^k= \sqrt{\det g_{mn}}\  D^k \; ,
\quad
{\cal F}^{kl}=   \epsilon^{klm} H_m \; , \quad H_m = \frac 12  \epsilon_{mkl} {\cal F}^{kl}
\, .
\end{equation}
 For linear (Maxwell) theory the Lagrangian density of the theory equals:
\begin{equation}\label{Lagr-Maxwell}
    L = - \frac 14 \sqrt{|\det g |} f_{\mu\nu} f^{\mu\nu} =\frac 12 \sqrt{|\det g |} \left( E^2 - B^2 \right) \, ,
\end{equation}
and, whence, ${\cal F}^{\nu\mu} = \sqrt{|\det g |} f^{\mu\nu}$ or, equivalently, ${\cal F} = *f$. Consequently, ``momenta'' are equal to ``velocities'': $D=E$ and $H=B$.

In absence of currents (i.e. when $J$=0), both the electric and magnetic fields play a symmetric role. This means that the Hodge-star operator ``$*$'' is an additional symmetry of the theory\footnote{In Lorentzian coordinates the Hodge operator ``*'' transforms: $E \rightarrow -B$ and $B \rightarrow E$. Similarly, $D \rightarrow -H$ and $H \rightarrow D$.} and we could, as well, begin with a potential $(C_\mu)=(C_0 , \vec{C})$ for the dual form $\textsl{h}= *  f$:
\begin{equation} \label{polaEB-star}
(\textsl{h}_{k0}) = -\vec{B}   = - \dot {\vec{C}}   +\vec{\nabla} C_0 \, , \quad \frac 12 \left( \epsilon^{mkl} \textsl{h}_{kl}\right) = \vec{E}  =   {\rm curl} \vec{C}     \, .
\end{equation}
Variational principle
\begin{equation}
\delta L(C_\nu , C_{\nu , \mu}) = \partial_\mu ({\cal H}^{\nu\mu} \delta
C_\nu)  =
(\partial_\mu {\cal H}^{\nu\mu}) \delta C_\nu +
{\cal H}^{\nu\mu} \delta C_{\nu , \mu} \ , \label{deltaL-ElmagC}
\end{equation}
of the same Lagrangian density
\begin{equation}\label{Lagr-Maxwell-h}
    L = - \frac 14 \sqrt{|\det g |} \textsl{h}_{\mu\nu} \textsl{h}^{\mu\nu} =\frac 12 \sqrt{|\det g |} \left( E^2 - B^2 \right) \, ,
\end{equation}
gives now the same field equations:
\begin{eqnarray}
\partial_\mu {\cal H}^{\nu\mu}  =  \frac {\partial L}{\partial C_\nu} = 0 \; ,
\qquad
{\cal H}^{\nu\mu}  =  \frac {\partial L}{\partial C_{\nu , \mu}}
= 2 \frac {\partial L}{\partial \textsl{h}_{\mu\nu}}
\ ,
\label{E-LpoleElmag-h}
\end{eqnarray}
with $D=E$ and $H=B$ playing a role of the corresponding canonical momenta ${\cal H} = * \textsl{h} = **f = -f$:
\begin{equation}\label{HD-C}
 {\cal H}^{0k}=- {\cal H}^{k0}=  -  \sqrt{\det g_{mn}}\,  H^k \; ,
\qquad
{\cal H}^{kl}=   \epsilon^{klm} D_m \; , \qquad D_m = \frac 12  \epsilon_{mkl} {\cal H}^{kl}
\, .
\end{equation}
The sum of \eqref{Lagr-Maxwell} and \eqref{Lagr-Maxwell-h} would imply the theory of two independent copies of  electromagnetic field, say $f$ and $\widetilde{f}$, such that $*\textsl{h}=\widetilde{f}$:
\begin{eqnarray}
  \delta L &=& \frac12\left[ (\partial_\mu {\cal F}^{\nu\mu}) \delta A_\nu + (\partial_\mu {\cal H}^{\nu\mu}) \delta C_\nu + {\cal F}^{\nu\mu} \delta A_{\nu , \mu}
  + {\cal H}^{\nu\mu} \delta C_{\nu , \mu} \right] \, . \label{due}
\end{eqnarray}
To have only one copy, we must impose constraint: ${\cal H} = *{\cal F}$. The constraint is equivalent to the requirement that $L$ depends only upon the sum ``$f+*\textsl{h}$'' and not upon the two potentials independently. Indeed, due to constraint we have:
\begin{eqnarray}
 {\cal F}^{\nu\mu} \delta A_{\nu , \mu}
  + {\cal H}^{\nu\mu} \delta C_{\nu , \mu} &=& {\cal F}^{\nu\mu} \delta A_{\nu ,\mu}
  +\left( * {\cal F}\right)^{\nu\mu} \delta C_{\nu , \mu} =
  {\cal F}^{\nu\mu} \delta \left( A_{\nu , \mu} + (* C)_{\nu , \mu} \right) \nonumber \\
    &=&   \frac 12 {\cal F}^{\nu\mu} \delta \left( f + * \textsl{h} \right)_{\mu\nu} \, . \label{F=f+*h}
\end{eqnarray}

Hence, for linear electrodynamics, we can take
\begin{eqnarray}
    L(A_\nu , C_\nu ,  A_{\nu , \mu}, C_{\nu , \mu} ) &=&  \frac 12 \sqrt{|\det g |} \left( E^2 - B^2 \right)  \nonumber \\
    &=&  - \frac 14 \sqrt{|\det g |} \left( f+ *\textsl{h} \right)_{\mu\nu} \left(  f +  *\textsl{h} \right)^{\mu\nu} \nonumber \\
    &=& - \frac 14 \sqrt{|\det g |} \left( {\rm d}A + *({\rm d} C) \right)_{\mu\nu} \left(  {\rm d}A + *({\rm d} C) \right)^{\mu\nu}  \, ,  \label{Lagr-Maxwell-hf}
\end{eqnarray}
which leads to a single copy of Maxwell electrodynamics with the Faraday tensor $\varphi := f+*\textsl{h}$ defined in terms of the two independent four-potentials $A$ and $C$:
\begin{equation}\label{fi}
    \varphi = {\rm d} A + * {\rm d} C \, .
\end{equation}
Moreover,
\begin{equation}\label{varfi}
   {\cal F}^{\nu\mu} =  - 2 \sqrt{|\det g |} \varphi^{\mu\nu}
\end{equation}
and
\[ L=  - \frac 14 \sqrt{|\det g |} \varphi_{\mu\nu} \varphi^{\mu\nu} \, . \]
Equation (\ref{fi}) in $(3+1)$-decomposition, reads:
\begin{equation} \label{pedyEMspin1}
\vec{E}   = - \dot {\vec{A}} + {\rm curl} \vec{C}  +\vec{\nabla} A_0 \, , \quad  \vec{B}  = \dot {\vec{C}} +  {\rm curl} \vec{A} -  \vec{\nabla} C_0   \, .
\end{equation}
Unlike in the standard variational formulation of electrodynamics: 1) the variation is performed with respect to two independent potentials: $A_\mu$ and $C_\mu$, and 2) the first pair of Maxwell equations is not imposed {\em a priori} but obtained from the variational principle. So, the complete set of Maxwell equations
\begin{eqnarray}
  {\rm div} D  &=& 0 \label{dDe}\\
  {\rm div} B &=& 0 \label{dBe}  \\
  \dot{D} &=& {\rm curl}\, B  \label{rBe}\\
  \dot{B} &=&  - {\rm curl} \, D  \label{rDe} \, ,
\end{eqnarray}
is {\em derived}, not {\em imposed} a priori.
Expressed in terms of potentials $(\vec A,\vec C, A_0, C_0)$, these equations read:
\begin{align}\label{potA}
\nabla\dot A_0 &= \ddot {\vec A} + \curl\curl \vec A \, ,\\ \label{potC}
  \nabla\dot C_0 &= \ddot {\vec C}+ \curl\curl \vec C \, .
\end{align}

The gauge group of such a theory is much bigger than the usual ``gradient gauge'': it is composed of all the transformations of the four-potentials which do not change the value of the field $\varphi$. Hence, not only ``$A \rightarrow A + {\rm d} \phi$'' and ``$C \rightarrow C + {\rm d} \psi$'', with two arbitrary functions $\phi$ and $\psi$ but, more generally, any transformation of the type
\begin{equation}\label{gauge-chi}
    A \rightarrow A + \xi \ \ \ ; \ \ \  C \rightarrow C + \eta \, ,
\end{equation}
where the four-covector fields $\xi= (\xi_\mu )$ and $\eta= (\eta_\mu )$ satisfy equation:
\begin{equation}\label{g-e}
    {\rm d} \xi + * {\rm d} \eta = 0 \, .
\end{equation}
It is obvious that both such ${\rm d} \xi$ and ${\rm d} \eta$ fulfill free Maxwell equations. In particular, the case ${\rm d} \xi = {\rm d} \eta = 0$ corresponds to the standard ``gradient gauge''.

We show in the sequel that, from the Hamiltonian point of view, such an exotic formulation of electrodynamics is perfectly equivalent to the standard formulation, using a single four-potential $(A_\mu)$.

\section{Hamiltonian Picture and Field Energy}

\subsection{Electromagnetic field energy in conventional formulation}\label{conventional}

Field energy is defined as the Hamiltonian function generating time evolution of the  field. To calculate its value, a $(3+1)$-decomposition has to be chosen and the Legendre transformation between ``velocities'' and ``momenta'' must be performed in the Lagrangian generating formula. In conventional formulation of electrodynamics we begin, therefore, with formula \eqref{deltaL-Elmag}:
\begin{eqnarray}
  \delta L &=& \partial_\mu ({\cal F}^{\nu\mu} \delta A_\nu)  =
\partial_0 ({\cal F}^{\nu 0} \delta A_\nu)  + \partial_k ({\cal F}^{\nu k} \delta A_\nu) \nonumber \\
    &=& \partial_0 ({\cal F}^{k 0} \delta A_k) + \partial_k ({\cal F}^{0 k} \delta A_0 + {\cal F}^{l k} \delta A_l) \nonumber \\
    &=& - \partial_0 ({\cal D}^{k} \delta A_k) + \partial_k ({\cal D}^{k} \delta A_0 + {\cal F}^{l k} \delta A_l) \nonumber \\
    &=& - {\dot{\cal D}}^{k} \delta A_k - {\cal D}^{k} \delta \dot{A}_k + \partial_k ({\cal D}^{k} \delta A_0 + {\cal F}^{l k} \delta A_l)
    \nonumber \\
    &=& \dot{A}_k \delta {\cal D}^{k} - {\dot{\cal D}}^{k} \delta A_k - \delta \left( {\cal D}^{k}  \dot{A}_k
    \right) + \partial_k ({\cal D}^{k} \delta A_0 + {\cal F}^{l k} \delta A_l)
    \, . \label{deltaL-Elmag-n}
\end{eqnarray}
Putting the complete derivative $\delta \left( {\cal D}^{k}  \dot{A}_k \right)$ on the left hand side, we obtain
\begin{equation}
- \delta \left(- {\cal D}^{k}  \dot{A}_k - L \right) = \dot{A}_k \delta {\cal D}^{k} - {\dot{\cal D}}^{k} \delta A_k
+ \partial_k ({\cal D}^{k} \delta A_0 + {\cal F}^{l k} \delta A_l)
 \ , \label{deltaL-Elmag-n1}
\end{equation}
which is analogous to the Hamiltonian formula $- \delta (p \dot{q} - L ) = \dot{p} \delta q - \dot{q} \delta p$ in mechanics, where $-\vec{\cal D}$ is the momentum canonically conjugate to $\vec{A}$ and ${\cal H} = - {\cal D}^{k}  \dot{A}_k - L$ is the Hamiltonian density. The boundary term $\partial_k ({\cal D}^{k} \delta A_0 + {\cal F}^{l k} \delta A_l)$ is usually neglected by sufficiently strong fall-off conditions at infinity. We stress, however, that the above symplectic approach enables one to localize energy within a (not necessary infinite) 3D volume $V$ with boundary $\partial V$. For this purpose we integrate \eqref{deltaL-Elmag-n1} over $V$ and obtain
\begin{equation}\label{H_V}
    - \delta {\cal H}_V = \int_V \left( \dot{A}_k \delta {\cal D}^{k} - {\dot{\cal D}}^{k} \delta A_k  \right) +
    \int_{\partial V} \left( {\cal D}^{\perp} \delta A_0 - {\cal F}^{\perp l} \delta A_l \right) \, ,
\end{equation}
where by ``$\perp$'' we denote the component perpendicular to the boundary and ${\cal H}_V = \int_{V}{\cal H}$. Imposing boundary conditions for $A_0$ and for $A_{\|}$ (components of $\vec{A}$ tangent to $\partial V$), we obtain an infinitely dimensional Hamiltonian system generated by the Hamiltonian functional equal to the ``Noether energy'' ${\cal H}_V$\footnote{The {\em time-time} component of the so called ``canonical'' energy-momentum tensor.}. Whereas controlling $A_{\|}$ at the boundary means to control ${\cal B}^\perp$, the control of the scalar potential $A_0$ means ``electric grounding'' of the boundary. This is {\em not} an adiabatic insulation of the field from the external World but rather a ``thermal bath'', with the Earth and its fixed scalar potential playing a role of the ``thermostat''. Hence, ${\cal H}_V$ is not the internal energy of the physical system: ``electro-magnetic field contained in $V$'', but rather its free energy: the uncontrolled flow of electric charges between $\partial V$ and the Earth plays the same role as the uncontrolled heat flow between the body and the thermostat during the isothermal processes. To avoid exchange of energy between the thermostat and the system, we must insulate it adiabatically. For this purpose we perform an extra Legendre transformation between ${\cal D}^{\perp}$ and $A_0$ at the boundary (cf. \cite{KIJ}):
\[
    {\cal D}^{\perp} \delta A_0 = \delta \left( {\cal D}^{\perp}  A_0 \right) - A_0 \delta {\cal D}^{\perp}
\]
and we obtain:
\begin{equation}\label{H_Vsym}
    - \delta \widetilde{\cal H}_V = \int_V \left( \dot{A}_k \delta {\cal D}^{k} - {\dot{\cal D}}^{k} \delta A_k  \right) +
    \int_{\partial V} \left(- A_0 \delta {\cal D}^{\perp}  - {\cal F}^{\perp k} \delta A_k \right) \, ,
\end{equation}
where
\begin{eqnarray}
  \widetilde{\cal H}_V &=& {\cal H}_V + \int_{\partial V} {\cal D}^{\perp}  A_0
  = \int_V - L - {\cal D}^{k}  \dot{A}_k + \partial_k \left( {\cal D}^{k}  A_0 \right) \\
    &=&  \int_V - L + {\cal D}^{k} \left( - \dot{A}_k + \partial_k   A_0 \right) = \int_V  {\cal D}^{k} E_k - L\, .
\end{eqnarray}
In linear Maxwell electrodynamics we obtain the standard, local, Maxwell energy density\footnote{The {\em time-time} component of the {\em symmetric} or Maxwell energy-momentum tensor.}:
\begin{equation}\label{e2+B2}
        - L + {\cal D}^{k} E_k = - \frac 12 \sqrt{|\det g |} \left( E^2 - B^2 \right) + \sqrt{|\det g |}  E^2 =
    \frac 12 \sqrt{|\det g |} \left( D^2 + B^2 \right) \, .
\end{equation}
The boundary term in \eqref{H_Vsym} vanishes if we control ${\cal D}^\perp$ and ${\cal B}^\perp$ on $\partial V$. Cauchy data are, therefore, described by : 1) electric induction $\vec{D}$ satisfying constraints \eqref{dDe}, and: 2) equivalence class of $\vec{A}$ {\em modulo} the gradient gauge $\vec{\nabla} A_0$ (each class uniquely represented by the magnetic field $\vec{B}$ satisfying constraints \eqref{dBe}). These Cauchy data form the phase space of the system equipped with the symplectic form
\begin{equation}\label{omega}
    \Omega = \int_V \delta A_k \wedge \delta D^k \, ,
\end{equation}
which is gauge-independent due to boundary conditions: $\delta D^\perp\left|_{\partial V} \right.= 0$. Due to this gauge-invariance, each class of equivalent field configurations can be uniquely represented by, i.e., the Coulomb-gauged potential $\widetilde{A}_k$ fulfilling the Coulomb gauge condition: ${\rm div} \, \widetilde{A} = 0$. Such a representant is unique if we impose the boundary condition $\delta \widetilde{A}^\perp\left|_{\partial V} \right.= 0$.
It can be proved that boundary conditions transform the Hamiltonian \eqref{e2+B2} into a genuine self-adjoint operator $\widetilde{\cal H}_V$, governing the field evolution on an appropriately chosen Hilbert-K\"ahler space of Cauchy data in $V$, and the symplectic form becomes:  $\Omega = \int_V \delta \widetilde{A}_k \wedge \delta D^k$.

\subsection{Phase space of Cauchy data}\label{Phase space}

The same conclusion may be obtained if we work directly with the field Cauchy data. To simplify notation, we use Lorentzian coordinates ($\sqrt{|\det g|} = 1$). According to \eqref{polaEB},  we have:
\begin{equation}\label{Lag-Cauchy}
    L = \frac 12 \left( \vec{E}^2 - \vec{B}^2 \right) = \frac 12 \left\{\left( \vec{\nabla} A_0 - \dot {\vec{A}}    \right)^2  - \left( {\rm curl} \vec{A} \right)^2 \right\} \, .
\end{equation}
We see that $A_0$ is a gauge variable because its momentum vanishes identically. Moreover, momentum canonically conjugate to 3D vector potential $\vec{A}$ equals:
\begin{equation}\label{DeqE}
    -\vec{D} := \frac {\partial L}{\partial \dot {\vec{A}}} = - \vec{E} \, .
\end{equation}
Consequently, variation of $L$ with respect to $A_0$ implies constraints:
\begin{equation}\label{constr}
   - \frac {\delta L}{\delta A_0} =  \partial_k D^k = 0 \, .
\end{equation}
Hence, we have:
\begin{eqnarray}
  \delta L &=&  D^k \delta \left( - \dot{A}_k + \partial_k A_0 \right)- B^k \delta \left(\epsilon_k^{\ ij} \partial_i A_j \right) = \\ \nonumber
  & & \hspace*{-1.5cm}
   \delta \left\{  D^k \left( - \dot{A}_k + \partial_k A_0 \right)\right\}   + \left( \dot{A}_k - \partial_k A_0 \right) \delta D^k + \partial_i \left( \epsilon^{ikj} B_k \delta A_j \right) - \left(\epsilon^{jik} \partial_i B_k
   \right) \delta A_j \, .
\end{eqnarray}
Putting the complete divergence $\delta \left(  D^k E_k \right)$ on the left-hand side, we obtain:
\begin{eqnarray}
  - \delta \left( D^k E_k -L \right)  &=& \dot{A}_k \delta D^k - \dot{D}^k \delta A_k + \partial_i
  \left( - A_0 \delta D^i + \epsilon^{ikj} B_k \delta A_j
  \right) \, ,
\end{eqnarray}
which finally implies \eqref{e2+B2} and \eqref{H_Vsym}. The boundary term vanishes if we control $D^\perp$ and $B^\perp = {\rm curl} A_{\|}$ on $\partial V$.

\subsection{Symplectic reduction in the FL formulation of electrodynamics}

In Fierz-Lanczos formulation we have more potentials, but also the gauge group (\ref{gauge-chi}--\ref{g-e}) is much bigger. In this Section we prove that -- when reduced with respect to constraints -- both formulations are perfectly equivalent. Hence, the Hamiltonian formulation and the notion of field energy does not depend upon a choice of a particular variational principle. Indeed, consider Lagrangian density \eqref{Lagr-Maxwell-h} and the corresponding Euler-Lagrange equations \eqref{E-LpoleElmag-h}:
\begin{eqnarray}
  {\rm div} \vec D  &=& 0 \label{divD}\\
  {\rm div} \vec B &=& 0  \label{divB} \\
  \dot{\vec D} &=& {\rm curl}\, \vec H  \label{dotD}\\
  \dot{\vec B} &=&  - {\rm curl} \, \vec E  \label{dotB} \\
  \vec D &=& \vec E \label{D=E} \\
  \vec H &=& \vec B \label{H=B}\;.
\end{eqnarray}
For fields satisfying these equations (i.e.~\textit{on shell}), integration by parts implies:
\begin{eqnarray}
  \delta \int_V L &=&  \int_V \left\{ \vec{D} \delta \left(- \dot {\vec{A}} + {\rm curl} \vec{C}  +\vec{\nabla} A_0 \right)
  - \vec H \delta \left(\dot {\vec{C}} +  {\rm curl} \vec{A} -  \vec{\nabla} C_0  \right) \right\} \\
   &=&  \int_V \left\{ - \vec{D} \delta  \dot{\vec{A}} + {\rm curl} \vec{D} \delta \vec{C} - \vec H \delta \dot {\vec{C}} - {\rm curl} \vec{H} \delta \vec{A}
   \right\} \\ \label{presymp}
   &=& - \int_V \left(\vec{D} \delta  \dot {\vec{A}}+ \dot {\vec{D}} \delta \vec{A} + \vec{H} \delta  \dot {\vec{C}}+ \dot {\vec{H}} \delta {\vec{C}}
   \right) = - \int_V \partial_0\left( \vec{D} \delta  {\vec{A}} + \vec{H} \delta  {\vec{C}} \right) \, .
\end{eqnarray}
Here, we have neglected the boundary integrals. They vanish because of appropriate boundary conditions which assure the adiabatic insulation of $V$.\footnote{The boundary conditions are necessary for the complete functional-analytic formulation of the Hamiltonian evolution. These issues (the appropriate definition of the Hilbert space of Cauchy data and the correct self-adjoint extension of the Hamiltonian) will be discussed in another paper.}
Hence, fields $\vec{D}$ and $\vec{H}$ play a role of (minus) momenta canonically conjugate to $\vec{A}$ and $\vec{C}$, respectively. To perform correctly Legendre transformation and obtain the value of the Hamiltonian function, we must reduce this symplectic structure to independent, physical degrees of freedom. For this purpose we use the Hodge decomposition of the space of three-dimensional vector fields $\vec X$ into two subspaces:
\be
\vec X = \vec X^v + \vec X^s\;,
\ee
where $\vec X^v$ is \textit{sourceless} (i.e. $\di\vec X^v = 0$) and $\curl\vec X^s = 0$. In particular, assuming trivial topology of the region $V$, we obtain that there exist a vector field $\vec W$ and a function $f$ such that  $\vec X^v = \curl \vec W$ and $\vec X^s = \vec \nabla f$.

Putting aside all the functional-analytic issues, consider field configuration having compact boundary in $V$. Integrating by parts, we see that $\vec X^v$ and $\vec X^s$ are mutually orthogonal\footnote{From the functional-analytic point of view the subspace of sourceless fields is defined as the $L^2$-closure of smooth, sourceless fields, having compact support in $V$ and the remaining subspace as its orthogonal complement in the Hilbert space $L^2$.}  in the Hilbert space $L^2$:
\[
(\vec X^v|\vec Y^s) = \int_V \vec X^v \cdot \vec Y^s = 0 \, .
\]

From (\ref{divD}--\ref{divB}) and (\ref{D=E}--\ref{H=B}) we have
\begin{align}
 \vec D^v &=  \vec{D} =  \vec{E} = \vec{E}^v \, ,\\
   \vec{H}^v &=  \vec{H} = \vec{B} = \vec{B}^v.
\end{align}
The sourceless parts of equations (\ref{potA}--\ref{potC}) imply wave equations for both $\vec A^v$ and $\vec C^v$.
 Define a sourceless vector potential $W$ for $C^v$, i.e. $\curl W = C^v$. Applying again the $\curl$ to this equation, we conclude that $\ddot W = -\curl\curl W$, i.e. $\square W = 0$.


Now, integrating by parts and using orthogonality relations, we reduce \eqref{presymp} as follows:
\begin{eqnarray*}
  \delta \int_V L &=& - \int_V \partial_0\left( \vec{D} \delta  {\vec{A}} + \vec{H} \delta  {\vec{C}} \right) = - \int_V \left(\vec{D} \delta  \dot {\vec{A}}+ \dot {\vec{H}} \delta {\vec{C}} + \vec{H} \delta  \dot {\vec{C}}
   + \dot {\vec{D}} \delta \vec{A}\right) \\
    &=&  -\int_V \left\{ \vec{D} \delta \dot {\vec{A}} - {\rm curl} \vec{D} \delta \vec{C} + \vec H \delta \dot {\vec{C}} + {\rm curl} \vec{H} \delta \vec{A}
   \right\} \\
   &=&  -\int_V \left\{\vec{D} \delta  \dot{{\vec{A}}}^v -\vec{D} \delta {\rm curl} \vec{C}^v +\vec H \delta \dot {\vec{C}}^v + \dot{\vec D} \delta \vec{A}^v
   \right\} \\
    &=&  -\int_V \left\{\vec{D} \delta  \dot{{\vec{A}}}^v -\vec{D} \delta \curl\curl \vec{W} + \vec H \delta \curl\dot {\vec{W}} + \dot{\vec D} \delta \vec{A}^v
   \right\} \\
    &=&  -\int_V \left\{ \vec{D} \delta  \dot{{\vec{A}}}^v +\vec{D} \delta \ddot {\vec{W}} +\curl\vec H \delta \dot {\vec{W}} +\dot{\vec D} \delta \vec{A}^v
   \right\} \\
    &=&  -\int_V \left\{ \vec{D} \delta  \dot{{\vec{A}}}^v +\vec{D} \delta \ddot {\vec{W}} + \dot{\vec D} \delta \dot {\vec{W}} + \dot{\vec D} \delta \vec{A}^v
   \right\} =
     -\int_V \left\{ \vec{D} \delta  (\dot{{\vec{A}}}^v+ \ddot {\vec{W}}) + \dot{\vec D} \delta (\vec{A}^v + \dot {\vec{W}}) \right\} \\
     &=& -\int_V \left( \vec{D} \delta  \dot{{\vec{\widetilde A}}} + \dot{\vec D} \delta \vec{\widetilde A}\right) = -\int_V \partial_0 \left( \vec{D} \delta  {{\vec{\widetilde A}}}\right) \, ,
\end{eqnarray*}
where we have defined the following, source-free, field: $\vec{\widetilde A} := \vec A^v + \dot{\vec W}$.

Hence, our original phase space $(\vec A,\vec C,\vec D,\vec H)$ of Cauchy data, equipped with a symplectic form $\omega = \delta \vec A\wedge\delta \vec D+ \delta \vec C\wedge\delta \vec H$, reduces {\em on shell} to $(\vec{\widetilde A}, \vec D)$ with a symplectic form $\widetilde\omega = \delta \vec{\widetilde A}\wedge\delta \vec D$, identical with the structure \eqref{omega} derived in Section \ref{conventional} from the conventional variational principle.

\subsection{Electromagnetic field energy in the FL formalism}

We see that the reduced (with respect to constraints) phase space  in Fierz-Lanczos formulation can be described by pair $(\vec{\widetilde A}, \vec D)$, where $\vec{\widetilde A} = \vec A^V + \dot{\vec W}$ plays a role of the field  configuration, whereas $- \vec D$ plays a role of its canonically conjugate momentum. It is, therefore equivalent to the corresponding phase space in the conventional formulation. Hence, Legendre transformation to the  Hamiltonian picture goes exactly as in Section \ref{conventional}:
\begin{align*}
H &= -L - \vec D\cdot\dot{\vec{\widetilde A}} = -\frac12(E^2-B^2) - \vec D\cdot(\dot{\vec A}^v + \ \ddot {\vec W}) = \frac12(B^2-D^2) - \vec D\cdot(\dot{{\vec A}}^v - \curl\curl \vec W) \\ &= \frac12(B^2-D^2) + \vec D(-\dot{{\vec A}}^v + \curl C^v) = \frac12(B^2-D^2) + \vec D\cdot \vec E^v = \frac12(D^2+B^2)\;,
\end{align*}
where we used the sourceless part of the first equation in \eqref{pedyEMspin1}: $\vec E^v = -\dot{{\vec A}}^v + \curl C^v $.

Reduction of the Fierz-Lanczos Lagrangian proposed in \cite{Cartin} (see our formula  \eqref{Lagr1}) can be obtained in a way entirely analogous to what was done above.

\subsection{Symplectic reduction of the spin-2 Fierz-Lanczos theory}

Take
\begin{eqnarray*}
  L &=& \frac 1{16} \sqrt{|\det g |} w_{\lambda\mu\nu\kappa} w^{\lambda\mu\nu\kappa} = \frac12 \sqrt{|\det g |} \left( D^2 - B^2 \right)  \\
    &=& \frac 12  \left\{ \left( \dot P - {\rm curl} \, S - \frac32 TS(\nabla a)\right)^2 - \left(  \dot S + {\rm curl} \,  P -\frac32 TS(\nabla b)\right)^2 \right\} \, .
\end{eqnarray*}
Euler-Lagrange equations (cf. \eqref{dD} and \eqref{pole4}) implied by $L$ read:
\begin{eqnarray}
  {\rm div} D  &=& 0 \label{divDs}\\
  {\rm div}  B &=& 0  \label{divBs} \\
  \dot{D} &=& {\rm curl}\, H  \label{dotDs}\\
  \dot{B} &=&  - {\rm curl} \, E  \label{dotBs} \\
 D &=& E \label{D=Es} \\
  H &=& B \label{H=Bs}\; .
\end{eqnarray}

For fields contained in a region $V$, satisfying proper boundary conditions, we can integrate $\delta L$ by parts and obtain \textit{on shell}:
\begin{eqnarray}
  \delta  \int_V L &=&  \int_V \left\{ {D} \delta \left(- \dot P + {\rm curl} S  +TS({\nabla} a) \right)
  -  H \delta \left(\dot S +  {\rm curl} P -  TS({\nabla} b)  \right) \right\} \\
   &=&  \int_V \left\{ - {D} \delta  \dot{P} + {\rm curl} {D} \delta S -  H \delta \dot S - {\rm curl} {H} \delta P
   \right\} \\ \label{presymps}
   &=& - \int_V \left({D} \delta  \dot {P}+ \dot {{D}} \delta P + \vec{H} \delta  \dot {S}+ \dot {\vec{H}} \delta {\vec{C}}
   \right) = - \int_V \partial_0\left( {D} \delta  {P} + \vec{H} \delta  {S} \right) \, .
\end{eqnarray}
Hence, fields $D$ and $H$ play a role of (minus) momenta canonically conjugate to $P$ and $S$, respectively. However, to perform correctly Legendre transformation and obtain Hamiltonian, we must reduce this symplectic structure to independent, physical degrees of freedom. For this purpose, we use decomposition of three-dimensional tensors of rank 2. Following Straumann (see \cite{Straumann}), an arbitrary 3D symmetric, traceless tensor $t_{kl}$ can be decomposed into three parts (called: {\em tensor}, {\em vector} and {\em scalar} parts, respectively):
\[
t_{kl} = t^t_{kl} + t^v_{kl} + t^s_{kl}\;,
\]
where
\begin{equation}
    \di t^t = 0\,, \quad
    {\rm tr}(t^t) = 0 \, ;\qquad
    t^v_{kl} = TS(\nabla \xi)_{kl}\,, \quad
    \di \xi = 0 \, ; \qquad t^s_{kl} = f_{,kl}-\frac13\Delta f
\end{equation}
for some function $f$ and a covector $\xi$.
For field configuration having compact boundary in $V$ (more generally: for fields fulfilling appropriate boundary conditions on $\partial V$), the decomposition is {\em unique} and the three components: $t^t$, $t^v$ and $t^s$ are mutually orthogonal with respect to the $L^2$-scalar product: $(t|s) = \int_V t\cdot s$.

From \eqref{divDs}-\eqref{divBs} and \eqref{D=Es}-\eqref{H=Bs} we have
\begin{equation}
 D^t =  {D} =  {E} = {E}^t \, , \qquad
   {H}^t =  {H} = {B} = {B}^t \, .
\end{equation}
By taking transverse-traceless part of equations \eqref{potP} and \eqref{potS}, we have that $P^t$ and $S^t$ fulfill wave equations. So, if we define $h$ as a tensor, such that
\be
\curl  h = S
\ee
than $h$ fulfills $\square h = 0$, too. (Existence and uniqueness of such $h$ is proved in Appendix A.) This equation is obviously equivalent to $\ddot h = -\curl\curl h$.

Now, we reduce expression~\eqref{presymps}, integrating by parts and using orthogonality relations:
\begin{align*}
     \delta L &= - \int_V \partial_0\left( {D} \delta  {P} + {H} \delta  {S} \right) =  -\int_V ({D} \delta  \dot {P} +\dot {{D}} \delta P + {H} \delta  \dot {S} + \dot {{H}} \delta {S}) \\
     &= -\int_V ({D} \delta  \dot {P}^t +\dot {{D}} \delta P^t + {H} \delta  \dot {S}^t + \dot {{H}} \delta {S}^t) \\ &= -\int_V ({D} \delta  \dot {P}^t +\dot {{D}} \delta P^t + \curl{H} \delta  \dot h - D \delta \curl{S}^t)\\ & = -\int_V ({D} \delta  \dot {P}^t +\dot {{D}} \delta P^t + \dot{D} \delta  \dot h - D \delta \curl\curl h) \\ &= -\int_V \left({D} \delta ( \dot {P}^t + \ddot h) +\dot {{D}} \delta (P^t +  \dot h)\right) =  -\int_V\left({D} \delta\dot{p} +\dot {{D}} \delta p\right) =-\int_V \partial_0 \left({D} \delta{p} \right) \,,
\end{align*}
where we denoted $p := P^t + \dot h$.
Hence, our symplectic structure $(P,S,D, H)$ with a symplectic form $\omega = \delta P\wedge\delta D+ \delta S\wedge\delta H$, became reduced to $(p, D)$ with a symplectic form $\widetilde\omega = \delta p\wedge\delta D$, derived in Section \ref{simple} from our naive variational principle (cf. \eqref{Omega_pD}).

\subsection{Field energy in the Fierz-Lanczos theory}

In this formulation the transition to the Hamiltonian picture is straightforward and gives results identical with the ones obtained in Section \ref{simple}. If $p =  P^t + \dot h$ is the configuration field, and $- D$ its canonical momentum then the Legendre transformation reads:
\begin{align*}
H &= -L -  D\cdot\dot{p} = -\frac12(E^2-B^2) - D\cdot(\dot{ P}^t + \, \ddot {h}) \\ &= \frac12(B^2-D^2) -  D\cdot(\dot{{ P}}^t - \curl\curl  h) =  \frac12(B^2-D^2) +  D(-\dot{{ P}}^t + \curl S^t) \\ & = \frac12(B^2-D^2) + D\cdot  E^t
 = \frac12(D^2+B^2)\;,
\end{align*}
where we have used the tensor part of the first equation in \eqref{weylodA}: $ \curl S^t -\dot{{P}}^t = E^t$.

\subsection{Poynting vector and energy flux in Fierz-Lanczos theory}

Similarly as in electrodynamics, the energy flux can also be localized. For this purpose we define the Poynting vector:
\begin{equation}\label{Poynting}
    {\cal S}^k = ( E \, \times \, B )^k := \epsilon^{klm} E_{li} B_m{^i} \, ,
\end{equation}
fulfilling the following identity:
\begin{eqnarray*}
  {\rm div} {\cal S} &=& \partial_k \left( \epsilon^{klm} E_{li} B_m{^i} \right) =
   \left( \epsilon^{klm} \partial_k E_{li} \right) B_m^{^i} +
     E_{li}  \left( \epsilon^{klm} \partial_k B_m^{^i} \right)
  \\
    &=&  \left(  \curl E \middle| B \right) - \left( E \middle| \curl B \right) =
    - \left(  \dot{B} \middle| B \right) - \left( E \middle| \dot{E} \right) =
    -  \partial_0  \left(\frac{E^2 + B^2}2\right) = - \dot{\cal H} \, ,
\end{eqnarray*}
equivalent to the  continuity equation:
\begin{equation}\label{cont}
    {\rm div} {\cal S} + \dot{\cal H} =0  \, .
\end{equation}
Integrating over any volume $V$, we obtain
\begin{equation}\label{fluxV}
    \dot{\cal H}_V = \frac {\rm d}{{\rm d}t} \int_V {\cal H} = - \int_{\partial V} {\cal S}^\perp \, .
\end{equation}
Hence, we are able to control the energy transfer through each portion of the boundary $\partial V$.

\section{Conclusions}
In this paper we were able to calculate the amount of energy $E_V$ carried by the massless spin-two field and contained within a space region $V \subset \mathbb{R}^3$. For this purpose we have used consequently definition of energy as the Hamiltonian function generating field evolution within $V$. {\em A priori}, evolution within $V$ is not unique because can be arbitrarily influenced by exterior of $V$. To make the system autonomous, we must insulate it adiabatically from this influence: appropriate conditions have to be imposed on the behaviour of the field at the boundary $\partial V$. Mathematically, control of boundary conditions select among possible self-adjoint extensions of the evolution operator (typically: the Laplace operator) a single one which is positive. Moreover, it enables us to organize the phase space of the field Cauchy data into a strong Hilbert-K\"ahler structure, where the ``well-posedness'' of the initial value problem is equivalent to the self-adjointness of the evolution operator. The use of specific representations of the theory (tensorial Fierz-Lanczos {\em versus} spinorial one, symplectic reduction by means of the Straumann decomposition {\em versus} imposing ``Coulomb gauge'' etc.) is irrelevant in this context: two such representations are isomorphic in a strong, functional-analytic sense. This way we have shown that the theory admits the ``local energy density'' $H = \frac{D^2 + B^2}2$ such that
\[
    E_V = \int_V H \, .
\]
Moreover, the flux of energy through boundary can also be localized by means of the Poynting vector \eqref{Poynting}.
We stress that -- contrary to the common belief -- such a local character of the field energy is rather exceptional. In particular, theories of gravitation (both the complete Einstein theory and its linearized  version) do not exhibit any such ``energy density''(or local flux represented by Poynting vector). Nevertheless, in both versions of the theory, energy $E_V$ and its flux can be uniquely defined by our procedure, even if the locality property \eqref{Esum} is not valid. The complete functional-analytic framework of our approach will be presented in the next paper.

\section*{Acknowledgements}

This research was supported in part by Narodowe Centrum Nauki (Poland) under Grant No. 2016/21/B/ST1/00940 and by the Swedish Research Council under grant no. 2016-06596 while JJ was in residence at Institut Mittag-Leffler in Djursholm, Sweden during the Research Program: General Relativity, Geometry and Analysis: beyond the first 100 years after Einstein, 02 September - 13 December 2019.

\appendix
\section{Existence of tensor potential for transverse-traceless tensors}

\begin{Lemma}
Given a symmetric, transverse-traceless field $B$ on a 3D-Euclidean space (i.e.~the Cauchy surface $\{ t = 0 \}$), there is a symmetric, transverse-traceless field $p$ such that
\begin{equation}\label{B=rotp1}
    B = {\rm curl} \ p \, .
\end{equation}
The field $p$ is implied by $B$ up to second derivatives $\partial_i \partial_j \varphi$ of a harmonic function: $\Delta \varphi =0$.
\end{Lemma}
\begin{proof}
Since for every $k=1,2,3$ the vector $B^{\bullet k}$ is divergence-free, we can solve equation ${\rm curl}\  a^{\bullet k} = B^{\bullet k}$. This means that there is a matrix $a_{ij}$ satisfying equation:
\begin{equation}\label{rota=B}
       \epsilon^{lij}\partial_i a_j^{\ k} = B^{lk} \, .
\end{equation}
Each solution is given uniquely up to a gradient. This means that for any triple $\phi^k$ of functions, the matrix
\[
    \widetilde{a}_j^{\ k} := a_j^{\ k} + \partial_j \phi^k \, ,
\]
is also  a solution of \eqref{rota=B}. To make the matrix $\widetilde{a}$ symmetric, we must fulfill three equations:
\begin{equation}\label{symm}
   0 = \epsilon^{njk} \widetilde{a}_{jk} = \epsilon^{njk} \left( {a}_{jk} + \partial_j \phi_k \right) \, ,
\end{equation}
or, equivalently
\begin{equation}\label{PhiPsi}
    {\rm curl} \, \vec{\phi} = \vec{\psi}  \, ,
\end{equation}
where we have defined vector fields $\vec{\phi} = \left( \phi^k \right)$ and $\vec{\psi} = \left( \psi^k \right)$, where $\psi^n := - \epsilon^{njk}  {a}_{jk}$.
A sufficient condition for the solvability is: ${\rm div}\, \psi = 0$. But, due to \eqref{rota=B}, we have:
\begin{equation}\label{divPsi}
    - {\rm div}\, \vec{\psi} = \partial_n \epsilon^{njk}  {a}_{jk} = \epsilon^{knj} \partial_n {a}_{jk} = B^k_{\ k} = 0 \, ,
\end{equation}
and, whence, the condition is fulfilled and the solution of \eqref{PhiPsi} is given uniquely, up to a gradient of a function, say $\varphi$. This means that $\phi_k$ is given uniquely up to $\partial_k \varphi$. We conclude that there is a solution of \eqref{rota=B} which is symmetric. It is given up to $\partial_j \partial_k \varphi$. This non-uniqueness can be used to make the solution traceless. For this purpose we put
\begin{equation}\label{p=a+ddfi}
    p_{ij} = \widetilde{a}_{ij} + \partial_i \partial_j \varphi \, ,
\end{equation}
and impose condition
\begin{equation}\label{tracep}
    0 = p_i^{\ i} = \widetilde{a}_i^{\ i} + \Delta \varphi \, ,
\end{equation}
which we solve for $\varphi$. This way we have $p$ which is another solution of \eqref{rota=B} and is: 1) symmetric and 2) traceless. But, it is also divergence-free because of the following identity:
\begin{eqnarray*}
  0= \epsilon_{nlk} B^{lk} &=& \epsilon_{nlk} \epsilon^{lij}\partial_i a_j^{\ k} = \left(\delta^i_k \delta^j_n - \delta^i_n \delta^j_k
  \right) \partial_i a_j^{\ k} = \partial_k a_n^{\ k} - \partial_n a_k^{\ k}\\
   &=&  \partial_k a_n^{\ k} \, .
\end{eqnarray*}
The Lemma is, therefore, proved and the solution $p_{ij}$ is given up to $\partial_i \partial_j \varphi$, where $\Delta \varphi = 0$.
\end{proof}

\section{Square of the Weyl tensor in (3+1)-decomposition}
Equalities \eqref{E-oraz-B}:
\[
E_{kl} =  w_{0k0l} \, , \quad
 \quad { B}_{ji} = \frac 12 \varepsilon_j^{\ kl}  w^{0}{_{ikl}}
 \]
imply also
\begin{equation}
    w^{0k0l} = E^{kl}\, , \quad w_{0kij} = -B_{kl}\varepsilon^l{_{ij}}\, , \quad w^{0kij} = B^{kl}\varepsilon_l{^{ij}}\, .
\end{equation}
Weyl property: $-\frac14\varepsilon^{\gamma\delta\alpha\beta}w_{\alpha\beta\mu\nu}\varepsilon^{\mu\nu\pi\rho} = w^{\gamma\delta\pi\rho}$ implies
\begin{equation}
w_{ijmn} = -\varepsilon_{ijk}E^{kl}\varepsilon_{mnl}.
\end{equation}
Finally, we obtain
\begin{align*}
  w_{\alpha\beta\mu\nu}w^{\alpha\beta\mu\nu} &= 4w_{0k0l}w^{0k0l} + 2w_{0kij}w^{0kij} + 2w_{ij0k}w^{ij0k} + w_{ijkl}w^{ijkl}=\\ &= 4E_{kl}E^{kl} -4\varepsilon^l{_{ij}}B_{kl}\varepsilon_m{^{ij}}B^{km} + \varepsilon_{ijm}E^{mn}\varepsilon_{kln}\varepsilon^{ija}E_{ab}\varepsilon^{klb} = \\
  &= 4E_{kl}E^{kl} -8B_{kl}B^{kl}+4E_{mn}E^{mn} = 8(E^2-B^2).
\end{align*}

\section{(3+1)-decomposition of the Lanczos potential}
If we define
\begin{eqnarray}
  P_{kl} &=& -A_{0(kl)} \\
  S_{kl} &=& -\frac12A_{ij(k}\varepsilon^{ij}{_{l)}} \\
  a_i &=& -A{_{0i0}}  \\
  b^i &=& -\frac 12 \varepsilon^{ikl} A_{kl0}   \ \  \Leftrightarrow  \ \   A_{ij0}=-b^m\varepsilon_{mij} \, ,
\end{eqnarray}
then we obtain
\begin{align}
    A_{0kl} &= A_{0(kl)} + A_{0[kl]} = -P_{kl} + \frac12(A_{0kl}-A_{0lk}) = \\
    &=-P_{kl} + \frac12(A_{0kl}+A_{lk0} + A_{k0l}) = -2_{kl} + \frac12A_{lk0} = -P_{kl} + \frac12b^m\varepsilon_{mkl} \, .
\end{align}
Tensor $A_{ij[k}\varepsilon_{l]}{^{ij}}$ is antisymmetric, so there exists a vector $c^m$ such that
\[
A_{ij[k}\varepsilon_{l]}{^{ij}} = c^m\varepsilon_{mkl} \, .
\]
Multiplying this equation by $\varepsilon^{klm}$, we have
\[
A_{ijk}(\eta^{im}\eta^{jk} - \eta^{ik}\eta^{jm}) =  2c^m,
\]
so
\[
c^m  = \frac12(A^{mj}{_j} - A^{jm}{_j}) = -A^{jm}{_j} = A^{0m}{_0} = -A_0{^m}{_0} = a^m.
\]
Now we decompose tensor $A_{ijk}\varepsilon^{ij}{_{l}}$ onto symmetric and antisymmetric part:
\begin{align}
    A_{ijk}\varepsilon^{ij}{_{l}} &= A_{ij(k}\varepsilon^{ij}{_{l)}}+A_{ij[k}\varepsilon^{ij}{_{l]}} = -2S_{kl} + a^j\varepsilon_{jkl}.
\end{align}
Multiplying this equality by $\varepsilon^l{_{mn}}$ leads to following result:
\begin{equation}
2A_{mnk} = -2S_{kl}\varepsilon^l{_{mn}}+a_{[m}\eta_{n]k}.
\end{equation}
Now, using \eqref{wodA}, we can express $E$ and $B$ in terms of $P$, $S$, $a$ and $b$:
\begin{align*}
    E_{kl} &= w_{0k0l} = A_{0kl;0} - A_{0k0;l}+A_{l00;k} - A_{l0k;0} - \left( A^i{_{00;i}}\eta_{kl} + A^0{_{(kl);0}}\eta_{00} + A^i{_{(kl);i}}\eta_{00}\right) \\&= -2\dot P_{kl} + 2a_{(k;l)} - a^i{_i}\eta_{kl} + \dot P_{kl} -  \varepsilon^{ji}{_{(k}}S_{l)j;i} + \frac12\left(a^i{_{;i}}\eta_{kl} - a_{(k;l)}\right) \\
    &= -\dot P_{kl} + (\curl S)_{kl} + \frac32a_{(k;l)} - \frac12a^i{_{;i}}\eta_{kl} \, ,
\end{align*}
\begin{align*}
    B_{kl} &= \frac12\varepsilon^{ij}{_l}w_{k0ij} = \frac12\varepsilon^{ij}{_l}\left(A_{k0j;i} - A_{k0i;j} + A_{ij0;k} - A_{ijk;0} - A^0{_{(0j);0}}\eta_{ki} - A^m{_{(0j);m}}\eta_{ki}  \right. \\
    &+ \left. A^0{_{(0i);0}}\eta_{kj} + A^m{_{(0i);m}}\eta_{kj}\right)  \\
    &= -\varepsilon^{ij}{_l}A_{0kj;i} - b_{l;k} + \dot S_{kl} - \frac12\varepsilon_{mkl}\dot a^m + \frac12\varepsilon^{i}{_{kl}}\dot a_i -\frac12\varepsilon^{i}{_{kl}} A^m{_{0i;m}} -\frac12\varepsilon^{i}{_{kl}} A^m{_{i0;m}}  \\
    &= \varepsilon^{ij}{_l}P_{kj;i} -\frac12 \varepsilon^{ij}{_l}\varepsilon_{kjm}b^m{_{;i}} - b_{l;k} + \dot S_{kl} + \varepsilon^{i}{_{kl}} P^m{_{i;m}} + \frac14\varepsilon^{i}{_{kl}}\varepsilon^{nm}{_i}b_{n;m} - \frac12\varepsilon^{i}{_{kl}}\varepsilon^{nm}{_i}b_{n;m} \\
    &= (\curl P)_{kl} -\varepsilon^i{_{kl}}P^m{_{i;m}}-\frac12b_{k;l} + \frac12b^i{_{;i}}\eta_{kl} - b_{l;k} + \dot S_{kl} + \varepsilon^{i}{_{kl}} P^m{_{i;m}} - \frac14b_{k;l} + \frac14b_{l;k} \\
    &= \dot S_{kl} + (\curl P)_{kl} - \frac32b_{(k;l)} + \frac12b^i{_{;i}}\eta_{kl} \, .
\end{align*}

\end{document}